%% file: main.tex
\documentclass[sigconf, nonacm]{acmart}

\hypersetup{pdfauthor={Prashant Kumar Pathak},
  pdftitle={Coverage Is Not Redundancy: Query-Aware Admission Indexes in Vector
  Databases Under Workload Drift}}

\usepackage{booktabs}

\begin{document}

\title{Coverage Is Not Redundancy: Maintenance Cost and Exposure of
Query-Aware Admission Indexes in Vector Databases Under Workload Drift}

\author{Prashant Kumar Pathak}
\affiliation{%
  \institution{Independent Researcher}
  \city{Milpitas}
  \state{CA}
  \country{USA}
}
\email{prashant.pathak@ieee.org}

\begin{abstract}
In a vector database serving retrieval at production scale, a single inserted document
can be retrieved for an anomalously large share of the query workload---a retrieval
\emph{hub}---and dominate the evidence returned for an entire topic. An emerging line of
work guards against this at ingest with an \emph{admission check} that rejects such a
document, maintaining a set of \emph{sentinel queries} and admitting a document only if
its reverse-$k$NN count $\kappa_S$ against them stays below a threshold $\tau$. Under
workload drift this sentinel set becomes a query-aware auxiliary index that must be
maintained online, and we study the cost that maintenance imposes on the ingest path. We
identify a structural limit---\emph{coverage is not redundancy}: a monitor stops
promoting sentinels once a region is \emph{covered}, but the predicate rejects a hub only
once $\tau$ sentinels \emph{witness} it, so exposure has an \emph{observation-limited}
floor $m_{\mathrm{safe}}+\lceil(\tau-c_0)_+/r\rceil$ that no reduction in update or
enforcement latency can close. On real HNSW, IVF-Flat, and IVF-PQ indexes over an
$8.8$M-vector MS MARCO corpus this floor is only a best case: as index recall falls,
exposure and churn rise above it and below recall $\approx0.5$ the gate stops containing
altogether---worst on the memory-compressed IVF-PQ used at billion scale---while a
recall-aware witness probe restores containment at a fixed $O(|S|d)$ admission cost,
under $0.1\%$ of the ANN insert. We validate the law under real (COVID-19) workload
drift, implement the gate in PostgreSQL/\texttt{pgvector} at a $0.33\%$ ingest tax, and
turn the bound into a provisioning rule that sizes the sentinel budget per emerging
region. A count test contains the hub where retrieval-time score normalizers (NNN,
QB-Norm) do not, and a pre-registered causal suite isolates the missing-coverage
mechanism from retrieval fragmentation across two embedding families (BGE-1024, E5-768).
Whether the resulting exposure crosses a given deployment budget is a calibrated workload
question; that an observation-limited floor exists, is latency-independent, and degrades
sharply with index recall is not.
\end{abstract}



\maketitle

\input{body}

\bibliographystyle{ACM-Reference-Format}
\bibliography{refs}

\end{document}

%% file: body.tex
\section{Introduction}

Vector databases back retrieval-augmented systems at production scale,
answering each query with the top-$k$ nearest documents from an HNSW- or IVF-indexed
corpus~\cite{hnsw,pq}. Because one document can be the nearest neighbor of many queries
at once, a single inserted document can come to be retrieved for a large fraction of a
topic's queries---a retrieval \emph{hub}~\cite{hubness}---and once it is, it dominates
the evidence the system returns for that entire topic. Operators meet this as a concrete
data-quality hazard, not an abstract one: an over-general FAQ page returned for every
question in a support corpus, a boilerplate or SEO-inflated page in an open web crawl,
or one tenant's promotional document surfacing across a shared multi-tenant store. In
each case a single document crowds out the specific evidence a query needs, and it does
so silently---nothing in a standard ANN index flags that one vector has become a hub.

The natural guard is an ingest-time \emph{admission check}: before a document is
committed to the store, test whether it would be retrieved for too large a share of the
workload and reject it if so. A practical realization keeps a set of \emph{sentinel
queries} $S$ drawn from the live workload and rejects a document $d$ once its
reverse-$k$NN count $\kappa_S(d)$---the number of sentinels that retrieve it---reaches a
threshold $\tau$ calibrated to a benign false-positive rate. As with any ingest-time
integrity check, the write that stresses it is the worst case: a document shaped to
dominate its target region, which we study exactly as a database reasons about
worst-case writes.\footnote{The same worst case has a security reading---adversarial
hubs and corpus poisoning~\cite{advhub,poisonedrag,corpuspoison}, studied for this gate
in concurrent work~\cite{pathakGating,pathakContainment}. Here the object is the
ingest-path maintenance cost.} On a \emph{stationary} workload the check is easy and
effective: a hub must lie near its region's query mass, which a representative $S$
already samples, so a high true reverse-$k$NN count and a low $\kappa_S$ cannot coexist.

Production workloads are not stationary. Query distributions drift, new topics emerge,
and the corpus grows, so $S$ is not fixed but a \emph{query-aware auxiliary index} that
must be maintained online as the workload moves. Keeping the admission check sound under
drift is therefore an incremental index-maintenance problem, and this paper studies its
cost on the ingest path---a cost that, we show, hides a structural limit.

That limit is a disagreement between two maintenance policies over ``enough.'' The
\emph{monitor} is novelty-driven---it stops promoting once a region is \emph{covered} by
one sentinel---while the \emph{predicate} is a count that rejects only once $\tau$
sentinels \emph{witness} the document. When an off-axis region emerges, an inserted hub
has near-zero count against the pre-existing sentinels, so the monitor marks the region
covered long before the predicate has the $\tau$ witnesses to reject it. That gap---the
\emph{missing coverage mass}---is the exposure the ingest path leaves open, and the
maintenance churn that closes it scales with the same quantity. It is what the rest of
the paper models, bounds, and evaluates.

\paragraph{Contributions.}
\begin{enumerate}
\item A \textbf{maintenance-cost model} for a query-aware admission index under
workload drift, driven by one quantity, the missing coverage mass $(\tau-c_0)_+$
(\S\ref{sec:model}--\S\ref{sec:coupling}).
\item An \textbf{observation-limited lower bound}
$E\ge m_{\mathrm{safe}}+\lceil(\tau-c_0)_+/r\rceil$ on exposure, independent of update
and enforcement latency (\S\ref{sec:bound}).
\item An \textbf{evaluation on real ANN indexes} at $8.8$M
(\S\ref{sec:ann}--\S\ref{sec:mitigation}): approximate retrieval inflates the floor and,
below a recall threshold, breaks containment---which a recall-aware probe restores.
\item A \textbf{\texttt{pgvector} prototype} and a provisioning rule
(\S\ref{sec:experiments}): the law holds under real (COVID) drift, the gate costs a
$0.33\%$ ingest tax, and the bound provisions the system.
\end{enumerate}

\paragraph{Deployment setting.} The examples above---open-crawl RAG pipelines,
multi-tenant stores, continuously updated knowledge bases---share the property that
makes admission control on the ingest path the natural guard: the corpus is not curated,
so a freshly inserted document can capture an emerging query region. In each, the
operator's question is a systems question: what does it cost the ingest path to keep an
admission check sound as the workload drifts, and what retrieval exposure remains while
it catches up? The takeaway from this paper is that the exposure floor and its
ingest-churn cost are a recurring cost to \emph{budget for} per emerging region, not a
one-off transient (\S\ref{sec:steady}), and that on an approximate index the floor
degrades sharply with recall (\S\ref{sec:ann})---so provisioning the sentinel budget,
$\tau$, and the serving index's recall are coupled decisions, not independent ones.

\section{System Model and Workload}
\label{sec:model}

We frame the system as a vector store with an ingest-time admission stage. The
store holds a corpus indexed for top-$k$ retrieval (an HNSW or IVF index); a query
workload arrives continuously; and on each document insertion an admission
predicate runs before the document is committed to the index.

\paragraph{Running example.} We ground the model in a concrete instance: a
customer-support corpus behind a retrieval-augmented assistant. As new product areas
appear, their questions form an \emph{emerging region} of the query workload; suppose an
over-general troubleshooting page---effectively ``have you tried restarting it?''---is
then ingested close enough to that whole region to be retrieved for most of its
questions (a \emph{hub}). Until the store has observed enough queries in the new area to
recognize the page as dominant, it is admitted and supplies the assistant's context for
nearly every question there. The quantities defined below name the pieces of this story:
the \emph{exposure} is how many such questions the page answers before it is caught, and
the \emph{maintenance cost} is what it takes the ingest path to catch it as the workload
keeps drifting.

\paragraph{Admission predicate.} Production vector stores are already moving toward
continuously-updated, ingest-heavy operation, where streaming index maintenance and
freshness are first-class concerns~\cite{milvus,freshdiskann,spfresh,dedrift}; an
ingest-time governance stage that bounds a document's retrieval dominance is a natural
next component for such systems, not a hypothetical one. We treat this admission check as
the \emph{query-aware auxiliary index} a designer would add for that purpose, and ask
what it costs to maintain. The store maintains an auxiliary set of
\emph{sentinel queries} $S$ sampled from the workload---a reverse-$k$NN probe set,
structurally the same kind of object as the secondary indexes vector stores already
keep, but keyed on queries rather than documents. For a document $d$, let
$\kappa_S(d)$ be its reverse-$k$NN count---the number of sentinels $q \in S$ for
which $d$ is in the top-$k$ retrieved for $q$. The predicate \emph{rejects} $d$ iff
$\kappa_S(d) \ge \tau$, with $\tau$ calibrated on benign documents to a target FPR.
Computing $\kappa_S(d)$ at insert is the admission cost; maintaining $S$ as the
workload drifts is the maintenance cost, and is what we analyze.

\paragraph{Maintenance under drift (the monitor).} As the query distribution
drifts, $S$ must track it or the predicate goes stale. A monitor detects an
emerging query region $R$ after $m_{\mathrm{safe}}$ qualifying detections~\cite{adwin} (a cluster
threshold itself calibrated FPR-safe) and promotes sentinels placed at observed
in-region query locations, evicting redundant sentinels to hold a budget. Each
promotion is an incremental index update; the count of promotions is the
maintenance \emph{churn} borne by the ingest path. The additive
$m_{\mathrm{safe}}$ term is a detector-design overhead, not the irreducible floor:
a sweep (\S\ref{sec:results}) measures the floor as
$p_{\mathrm{cover}}\,m_{\mathrm{safe}}+\mathrm{deficit}$ (slope $1.00$, intercept
$=\mathrm{deficit}$), so the contribution-bearing, irreducible term is the deficit
$\lceil(\tau-c_0)_+/r\rceil$, bounded by the admission-FPR that fixes $\tau$.
Reducing $m_{\mathrm{safe}}$ trades against the monitor's false-region rate, which
binds only for regions overlapping the benign workload---for the well-separated
regions where windows open in our data it is itself reducible---whereas the
deficit term binds regardless.

\paragraph{Re-evaluation and revocation.} A document admitted before region $R$'s
queries arrive sits in the corpus with $\kappa_S < \tau$. Later promotions raise its
count, but admission is checked at insert, so a stale-at-insert document is
contained only if the store re-evaluates and revokes admitted documents after
sentinel updates. We thus decompose exposure
$E = E_{\mathrm{obs}} + E_{\mathrm{rev}}$: $E_{\mathrm{obs}}$ is the
observation-limited component (witness acquisition), and $E_{\mathrm{rev}}$ is the
\emph{re-evaluation overhead}---the cost of re-checking $\kappa_S\ge\tau$ against
updated sentinels and acting on the result, whether by revoking an admitted document
or by suppressing it at serve time. It is tunable by re-evaluation cadence: periodic
revocation is a coarse cadence, and a serve-time count gate (\S\ref{sec:placement})
is the maximum-cadence limit that re-evaluates on every query. Cadence is therefore
an engineering knob, not part of the floor---every cadence pays the \emph{same}
$E_{\mathrm{obs}}$, which is what our results concern.

\paragraph{Workload assumption (worst-case insertion).} The worst-case insertion is a
concept \emph{hub} $h$ targeting an emergent region $D$ at angular displacement
$\theta$ from the established query centroid. At insertion $S$ covers only
established regions, so the initial support $c_0(h) = \kappa_{S_0}(h)$ is small; in
the off-axis regime $\theta$ large, $c_0(h) < \tau$ (the geometric condition under
which a window opens---this is the only role the embedding geometry plays).
Exposure is the number of distinct in-region queries that retrieve $h$ before it is
contained. The benign version is equally important: an emerging query region
whose nearest document is an under-covered hub arises under \emph{natural} hubness
and drift~\cite{hubness,hubsurvey} with no adversary at all---a newly popular
document in a freshly drifted region is structurally the same object. The floor is
therefore a cost the maintenance must pay whenever a region emerges, whether or not
one believes the adversarial insertion.

\section{The Witness-Deficit Lower Bound}
\label{sec:bound}

\begin{definition}[Missing coverage mass / witness deficit]
For a document $h$ with initial support $c_0(h)$ against threshold $\tau$, the
\emph{missing coverage mass} is $d(h) = (\tau - c_0(h))_+$: the number of additional
sentinel witnesses the admission index must acquire in $h$'s region before the
predicate can reject $h$. We use ``missing coverage mass'' (the cost-model reading)
and ``witness deficit'' (the mechanism reading) interchangeably.
\end{definition}

\begin{lemma}[Covering implies exposing]
\label{lem:cover}
Let $\rho_{\mathrm{cov}}$ be the angular radius within which an observed query, once
promoted to a sentinel, covers $h$ (places $h$ in that sentinel's top-$k$, so the
promotion raises $\kappa_S(h)$), and let $\rho_{\mathrm{ret}}$ be the radius within
which a query retrieves $h$ ($h$ in its own served top-$k$). If
$\rho_{\mathrm{cov}} \le \rho_{\mathrm{ret}}$ then every promotion that raises
$\kappa_S(h)$ is seeded by a query that also retrieves $h$; equivalently the
per-observation cover probability is $p_{\mathrm{cover}}=1$.
\end{lemma}

\begin{proof}
A promotion raises $\kappa_S(h)$ only if its seeding query $q$ lies within
$\rho_{\mathrm{cov}}$ of $h$ (otherwise the promoted sentinel does not retrieve $h$).
By hypothesis $\rho_{\mathrm{cov}}\le\rho_{\mathrm{ret}}$, so $q$ lies within
$\rho_{\mathrm{ret}}$ of $h$ and hence retrieves $h$ in its own served top-$k$.
\end{proof}

\paragraph{Assumptions.} The bound below holds under a stated maintenance model:
\emph{(1)}~containing $h$ requires $\kappa_S(h)\ge\tau$---a single sentinel does not
suffice; \emph{(2)}~a sentinel may be promoted only at an observed query location, with
no oracle pre-placing one on an unseen target; and \emph{(3)}~a promotion raises
$\kappa_S(h)$ by at most $r$. Assumption~2 makes the floor observation-limited; 1 and 3
fix the witness arithmetic.

\begin{proposition}[Observation-limited floor]
\label{prop:floor}
Under Assumptions 1--3 with re-evaluation/revocation semantics, the number of
documents exposed (distinct in-region queries retrieving $h$) before containment
satisfies
\[
E \;\ge\; p_{\mathrm{cover}}\,m_{\mathrm{safe}} + \left\lceil \frac{(\tau - c_0(h))_+}{r} \right\rceil,
\]
which under Lemma~\ref{lem:cover} ($p_{\mathrm{cover}}=1$, the emergence
regime) is the tight floor $m_{\mathrm{safe}} + \lceil(\tau - c_0(h))_+/r\rceil$.
The right-hand side contains no update- or enforcement-latency term.
\end{proposition}

\begin{proof}[Proof sketch]
Containment requires $\kappa_S(h)$ to rise from $c_0$ to $\tau$, i.e. an additional
count of $(\tau - c_0)_+$. Each promotion contributes at most $r$, so at least
$\lceil(\tau - c_0)_+/r\rceil$ promotions are required. Promotion in region $R$
cannot begin until the monitor's $m_{\mathrm{safe}}$ detections fire; each such
detection is an in-region query that, under Lemma~\ref{lem:cover}, retrieves $h$
\emph{before} any sentinel exists to reject it, and is therefore an exposed
retrieval (where $p_{\mathrm{cover}}<1$ this term relaxes to
$p_{\mathrm{cover}}\,m_{\mathrm{safe}}$, as \S\ref{sec:results} measures). Each subsequent
promotion is gated on a fresh observed in-region query, which by
Lemma~\ref{lem:cover} retrieves $h$ while $\kappa_S(h) < \tau$ (the bound is thus
\emph{conditional} on Lemma~\ref{lem:cover}; we verify its hypothesis, $p_{\mathrm{cover}}=1$,
empirically below). Hence at least
$m_{\mathrm{safe}} + \lceil(\tau - c_0)_+/r\rceil$ in-region queries retrieve $h$
before containment. Finally, none of these counts depends on wall-clock latency:
even with instantaneous detection, update, and enforcement, witnesses cannot be
synthesized faster than in-region queries arrive, because a sentinel is promotable
only at an observed query location.
\end{proof}

\begin{corollary}[Latency-independence]
\label{cor:latency}
$\lim_{T_{\mathrm{update}},T_{\mathrm{enforce}}\to 0} E \ge
\lceil(\tau - c_0)_+/r\rceil > 0$ whenever $c_0 < \tau$ (the deficit term binds
regardless of $p_{\mathrm{cover}}$). The
floor is an evidence-acquisition limit, distinct from processing delay.
\end{corollary}

This is the formal content of ``coverage is not redundancy'': the monitor's
stopping condition ($\ge 1$ sentinel) is structurally weaker than the predicate's
rejection condition ($\ge \tau$ witnesses), and the gap---the missing coverage mass
$d(h)$---is paid down one observed query at a time. Maintenance cost and exposure
both scale linearly in it. Lemma~\ref{lem:cover}'s hypothesis is precisely the
\emph{emergence regime}: with the off-axis region empty of competitors a
query's served top-$k$ admits $h$ exactly when $h$ is near it, so
$\rho_{\mathrm{cov}}=\rho_{\mathrm{ret}}$ and $p_{\mathrm{cover}}=1$ (measured in
\S\ref{sec:results}). As legitimate documents accumulate $\rho_{\mathrm{ret}}$
shrinks below $\rho_{\mathrm{cov}}$, $p_{\mathrm{cover}}$ falls, and the exposure
term weakens---the density boundary of \S\ref{sec:dense}. The bound is thus a
statement about the emergence window resting on one geometric hypothesis we
\emph{verify} rather than assume.

\section{Cost--Exposure Coupling}
\label{sec:coupling}

Each of the $\lceil(\tau - c_0)_+/r\rceil$ promotions that close the gate is a
sentinel addition---one unit of \emph{churn}. By the same argument, each is also
seeded by an in-region query that, until closure, retrieves $h$---one unit of
\emph{exposure}. Covering a hub and retrieving it are the same $k$NN event.
Therefore:
\[
\underbrace{\text{churn-to-close}}_{\text{promotions}}
=
\underbrace{\text{exposure-to-close}}_{\text{exposed retrievals}}
= m_{\mathrm{safe}} + \left\lceil \frac{(\tau - c_0)_+}{r} \right\rceil.
\]

This identity holds while the region is \emph{sparse}, i.e.\ the per-observation
cover probability $p_{\mathrm{cover}}\approx 1$---the emergence window in which the
worst-case insertion places the hub. Once legitimate in-region documents accumulate,
$p_{\mathrm{cover}}$ falls and the two sides separate (churn grows as
$\sim\!d/p_{\mathrm{cover}}$ while exposure does not); we measure this boundary on
real embeddings in \S\ref{sec:dense}.

\begin{corollary}[Budget-parametric infeasibility]
\label{cor:triad}
Under Assumptions 1--3, for any churn budget
$B_{\mathrm{churn}} < m_{\mathrm{safe}} + \lceil(\tau - c_0)_+/r\rceil$ there is no
exposure budget $B_{\mathrm{exp}}$ under which a containing novelty-only maintenance
policy is feasible: closing the gate necessarily spends that much churn, and not
closing it leaves exposure to grow with continued in-region traffic.
\end{corollary}

This is stronger than ``no configuration satisfied our particular budget box.'' It
states a \emph{frontier}: the deficit lower-bounds the churn required to contain,
for any exposure target. A novelty-only policy can hold churn near 1 (stop after one
cover) only by leaving exposure to accumulate; it can drive exposure to the floor
only by paying churn equal to the deficit. Within this model, there is no low-cost,
low-exposure corner.

\section{Evaluation Methodology}
\label{sec:method}

Each experiment is a control that isolates one explanation for the exposure floor from
the alternatives; thresholds, grids, and decision rules are fixed in advance. The
central idealization permits sentinel promotion \emph{only at observed query
locations}, at zero latency and cost. This is deliberately generous to the maintenance
policy: an oracle that instead pre-placed a sentinel on the unseen target would define
the observation floor away, so we exclude it, and the floor of
Proposition~\ref{prop:floor} is what survives under the fair idealization.

Each gate eliminates one alternative explanation before the next runs. The central
comparison is a three-model contrast---an \emph{oracle} (instant sufficient
coverage), \emph{continued promotion} (keep promoting past coverage), and
\emph{novelty-stopping} (stop at one cover)---which isolates the cost attributable
to stopping at coverage.

\paragraph{Setup.} Encoders: \texttt{bge-large-en-v1.5} (1024-d) and
\texttt{e5-base-v2} (768-d), unit-normalized, exact cosine $k$NN. Corpora:
BEIR-derived collections (a $\sim$100k-document reconstruction). Drift is realized
by a synthetic rotation control giving a clean monotone $\theta$ sweep. Operating
point: $\tau = 14$ (FPR $10^{-3}$) on BGE, $\tau = 18$ on E5;
$m_{\mathrm{safe}} = 5$. The worst-case insertion is a concept hub for the emergent
region.

\paragraph{Evaluation tiers.} The evaluation is deterministic. The synthetic controls
run on the $\theta$ sweep and replicate cross-family on E5; the real-embedding
experiments---the dense-region competitor-density control (\S\ref{sec:dense}) and the
real-domain three-model contrast (\S\ref{sec:results}) on the $\sim$100k-document
BGE/E5 reconstruction---are evaluated separately.

\section{Experimental Validation}
\label{sec:results}

The controls below hold on BGE and replicate identically on E5. Each isolates one
explanation for the exposure floor---the lower bound itself, the deficit versus
retrieval fragmentation, the scaling law, the coverage--redundancy ordering, and the
comparison against retrieval-time normalization---before the next.

\subsection{Deficit versus fragmentation}

The leading alternative to the deficit mechanism is \emph{retrieval fragmentation}:
that off-axis, dispersed regional queries scatter witnesses across nearby documents
so that no single target accumulates $\tau$, independent of any
coverage--redundancy mismatch. We isolate the two (Figure~\ref{fig:frag}). In the
measured geometry, the per-observation cover probability $p_{\mathrm{cover}} = 1.0$
across regional query-dispersion spreads $5^\circ$--$40^\circ$ (the in-region
query spread, distinct from the rotation $\theta$ of \S\ref{sec:phase}): where the
deficit exists (off-axis), the
region is empty of corpus competitors, so each witness covers the hub. In a
counterfactual that injects fragmentation as a free parameter
($p_{\mathrm{cover}} \in \{1.0, 0.7, 0.4\}$), the deficit slope is unchanged
(exposure stays linear in $d$ with slope 1.0); injected fragmentation shifts only the
fit scatter, never the slope. The reason
is structural: covering a hub and retrieving it
are the same $k$NN event, so dispersion changes the number of \emph{total}
observations but not the \emph{exposed-before-closure} count.

\begin{figure}[t]
\centering
\includegraphics[width=\columnwidth]{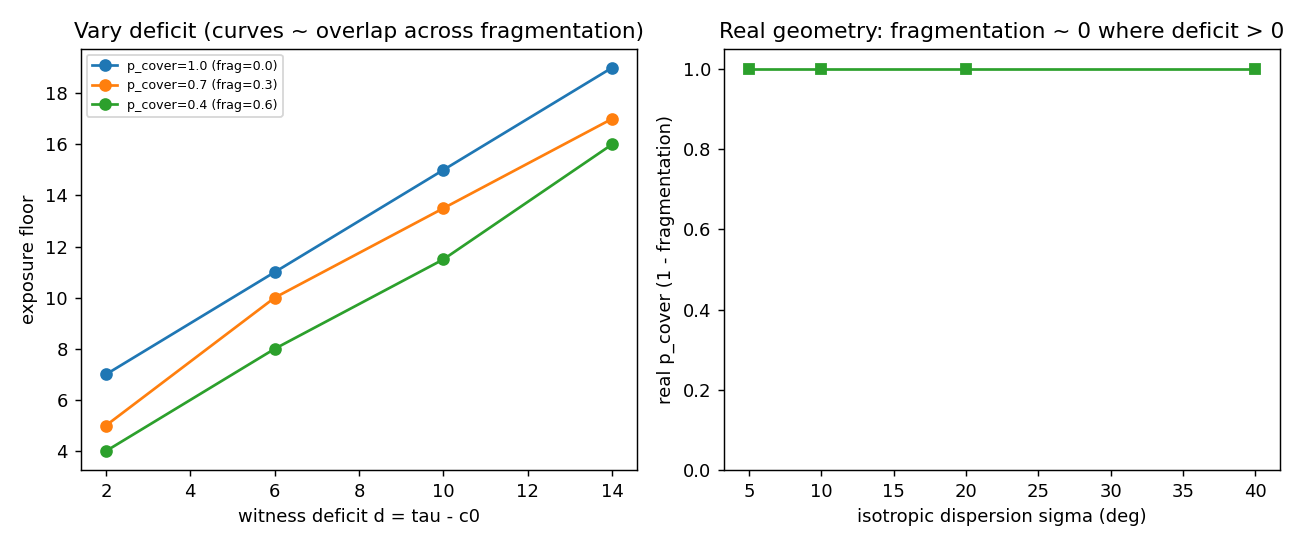}
\Description{Two panels showing exposure is linear in the witness deficit and invariant to injected fragmentation, and that real per-observation cover probability is one at all dispersions.}
\caption{Deficit versus fragmentation. Left: exposure is linear in the witness deficit
$d$ with slope 1.0, and the three fragmentation curves
($p_{\mathrm{cover}}{\in}\{1.0,0.7,0.4\}$) nearly overlap---fragmentation shifts only
a bounded detection overhead, never the floor. Right: in the \emph{sparse /
emergence regime} (off-axis region empty of competitors) $p_{\mathrm{cover}}{=}1.0$
at every dispersion, so fragmentation is structurally absent where the deficit
lives. Exposure is caused by the witness deficit, not fragmentation. (The dense
regime, where real documents make $p_{\mathrm{cover}}$ fall, is \S\ref{sec:dense}.)}
\label{fig:frag}
\end{figure}

Forcing fragmentation in as a free parameter and finding the deficit law invariant
rules it out in the sparse regime. In the dense regime (\S\ref{sec:dense}),
fragmentation \emph{is} a co-mechanism---it leaves the deficit intact but decouples
churn from exposure---so the deficit governs the emergence window, with a measured
density boundary beyond it.

\subsection{The scaling law}

Sweeping $\tau$ gives $E$ linear in the deficit with unit slope and $R^2=1.0$
(Figure~\ref{fig:law}); varying the witnesses-per-observation $r\in\{1,2,4\}$ moves the
floor inversely and exactly---median exposure $\{1{:}19,2{:}12,4{:}9\}$, fitting
$E=\lceil(\tau-c_0)_+/r\rceil+m_{\mathrm{safe}}$. Manipulating the term the bound says
controls the floor produces the predicted response.

This fit is partly definitional: the cost--exposure identity (\S\ref{sec:coupling})
makes churn and exposure the same $k$NN event, so it confirms internal consistency
rather than an independent prediction. The evidence that could have disagreed comes
from the real-embedding tests below---the four-domain floor predicted from an
\emph{independently measured} $p_{\mathrm{cover}}$ (\S\ref{sec:placement}) and the
per-region recurrence over $24$ sub-regions (\S\ref{sec:steady}).

\begin{figure}[t]
\centering
\includegraphics[width=\columnwidth]{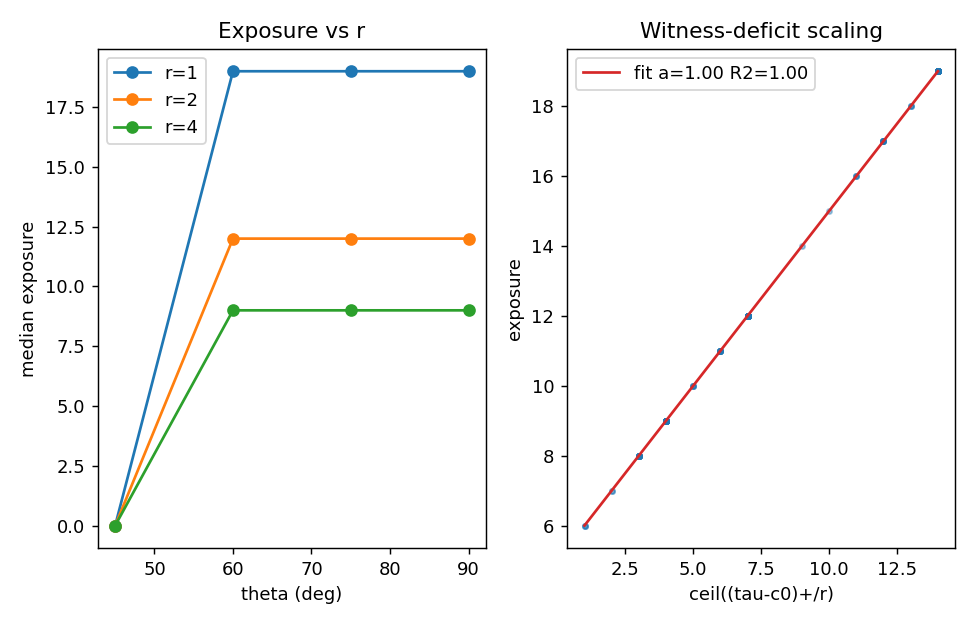}
\Description{Two panels showing median exposure decreasing with witnesses-per-observation r and collapsing onto the predicted linear law with R-squared one.}
\caption{The scaling law ($r$-sweep). Left: median exposure falls as the
witnesses-per-observation $r$ rises ($\{1{:}19,2{:}12,4{:}9\}$). Right: exposure
collapses onto $E = a\lceil(\tau-c_0)_+/r\rceil + b$ with $a{=}1.0$, $R^2{=}1.0$.}
\label{fig:law}
\end{figure}

\subsection{The coverage--redundancy ordering}

The three-model contrast yields exposure 0 (oracle) and 19 (continued
promotion)---both containment points, at which exposure stops accruing once the hub
is contained---against 248 for novelty-stopping. The 248 is not a containment
ceiling: a novelty-stopping monitor halts at one cover, below $\tau$, so it never
contains, and its exposure accrues for the lifetime of the in-region stream, growing
without bound; 248 is merely the value at this suite's run horizon (as 1000 is on the
real domain below). The gap is therefore categorical, not numerical---promoting to
redundancy contains at $m_{\mathrm{safe}}+\mathrm{deficit}$ while stopping at coverage
never does---and the oracle's 0 confirms it is the coverage--redundancy mismatch, not
an artifact.

\emph{This ordering---the paper's title---holds on real geometry, not only the
synthetic sweep.} On the real biomedical domain of \S\ref{sec:dense} (BGE,
sparse/emergence regime, real established sentinels giving $c_0 = 2$, real in-region
queries driving containment), the three models give exposure
$0$ / $17$ / $1000$ (oracle / continued / novelty). Continued promotion contains at
exactly $m_{\mathrm{safe}}+\mathrm{deficit}=17$---the law---while novelty-stopping
stalls at $c_0{+}1 = 3 < \tau$, never contains, and accrues exposure for the
lifetime of the in-region stream (here capped at $1000$ arrivals; it grows without
bound with continued traffic). The coverage--redundancy ordering and the law are
thus grounded on real embeddings, not inferred from the synthetic rotation.

\subsection{Phase structure and cross-embedding}
\label{sec:phase}

Here $\theta$ is the \emph{synthetic-rotation knob}: we rigidly rotate a concept
hub about the established query centroid to sweep the deficit through a clean,
monotone control variable. The floor vs.\ $\theta$ is $[0, 0, 0, 0, 19, 19, 19]$
for $\theta = 0^\circ \ldots 90^\circ$: illumination-safe (floor 0, $c_0 \ge \tau$)
below a rotation threshold ($\theta \le 45^\circ$ on BGE), full deficit above it
($\theta \ge 60^\circ$; Figure~\ref{fig:phase}). The same qualitative structure
holds on E5 with the threshold shifted and floor $23 = m_{\mathrm{safe}} + 18$
(vs.\ BGE $19 = m_{\mathrm{safe}} + 14$; at full rotation $c_0{=}0$ so the deficit
equals $\tau$), tracking the recalibrated $\tau$
(Table~\ref{tab:cross}); the deficit scaling replicates on both. We stress what
this is and is not: the rotation \emph{constructs} a monotone angle$\to$deficit
map purely to vary the deficit as an instrument; the specific angle thresholds are
a property of that controlled rotation, not a claim that real domains below some
angle are safe. Real coverage is anisotropic and is \emph{not} a function of angle
alone---\S\ref{sec:dense} grounds this, where a real drifted domain at only
$34^\circ$ already carries a deficit. The transferable result is the mechanism
$E=f(\mathrm{deficit})$, swept here by $\theta$ and grounded on real geometry in
\S\ref{sec:dense}; it is not the phase angles.

\begin{figure}[t]
\centering
\includegraphics[width=\columnwidth]{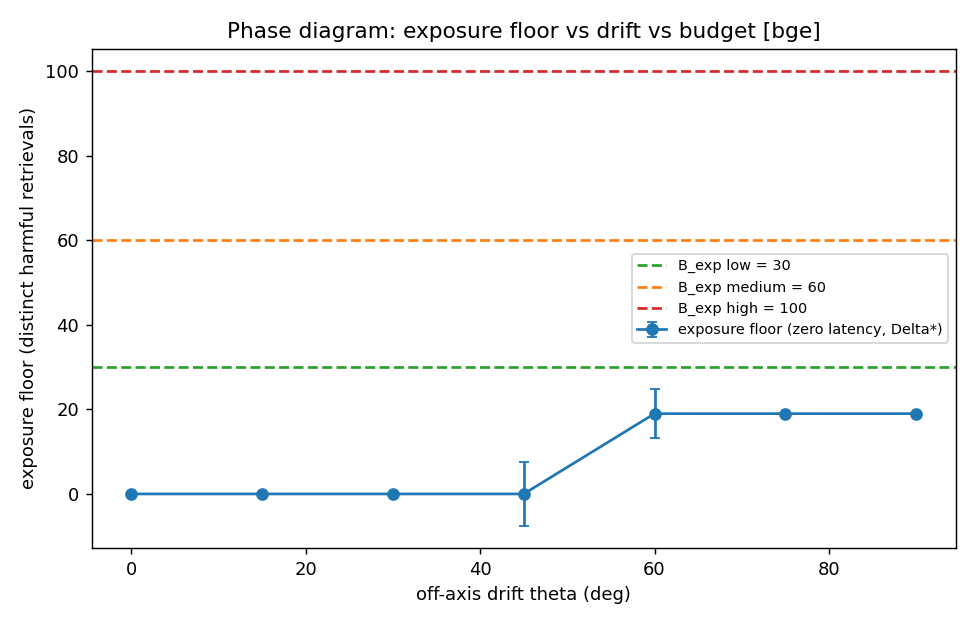}
\Description{Exposure floor versus off-axis drift angle theta rising from zero to the full deficit, shown against three placeholder budget bands.}
\caption{Phase structure (BGE). The exposure floor rises from 0 (illumination-safe,
$c_0\ge\tau$) to the full deficit as $\theta$ grows, against the three placeholder
budget bands. The transition is in the floor; whether it crosses a budget is a
calibrated-workload question (\S\ref{sec:verdict}).}
\label{fig:phase}
\end{figure}

\begin{table*}[t]
\caption{Cross-family replication. Floor $= m_{\mathrm{safe}}+$deficit; magnitude
scales with the recalibrated $\tau$ (BGE 14, E5 18).}
\label{tab:cross}
\small
\begin{tabular}{@{}lcccc@{}}
\toprule
Encoder (dim) & low-$\theta$ safe & high-$\theta$ opens & floor & deficit $R^2$ \\
\midrule
BGE-large (1024) & $\theta \le 45^\circ$ & $\theta \ge 60^\circ$ & 19 & 1.00 \\
E5-base (768) & $\theta \le 30^\circ$ & $\theta \ge 45^\circ$ & 23 & 1.00 \\
\bottomrule
\end{tabular}
\end{table*}

\subsection{Dense-region analysis}
\label{sec:dense}

The fragmentation control rules it out only where the off-axis region is empty of
competitors. We supply the missing control on real embeddings: we place the hub at
a real emergent-domain centroid (BGE; $\theta = 34^\circ$ from the established
query centroid) and let the \emph{real} in-region documents be the competitors,
sweeping the fraction present, while \emph{real} in-region queries drive
containment (Figure~\ref{fig:dense}). Note first what this real geometry says about
angle: the emergent domain sits at only $\theta = 34^\circ$ yet its initial support
is $c_0 = 2$, a deficit of $12$---a window is open at an angle the
\emph{synthetic} rotation of \S\ref{sec:phase} would still call illumination-safe
($\theta \le 45^\circ$). The two do not conflict; they measure different things.
The rotation moves a hub rigidly about the centroid, so coverage falls smoothly
with angle by construction; the real domain is a distinct cluster (biomedical
queries) sitting in a direction the established sentinels---which cluster in three
sub-domains, not a sphere---happen not to cover, so it is poorly covered despite a
moderate angle. Real coverage is anisotropic: $c_0$ is a property of \emph{where}
the domain falls relative to the sentinel clusters, not of its angle. This is
precisely why we use $\theta$ only as a controlled knob to sweep the deficit
(\S\ref{sec:phase}) and ground the real angle$\to$coverage relation here. With the
real $c_0 = 2$ the law predicts a floor of $m_{\mathrm{safe}} + 12 = 17$, which is
what we observe in the sparse regime. The result then locates a sharp regime
boundary. ``Density'' here is the fraction of the emergent domain's own documents
already present in the region (out of $25{,}000$ domain-D documents); the sweep is
that fraction. The emergence instant genuinely sits in the sparse corner: a hub
inserted when a region \emph{first} appears precedes that region's legitimate
documents---which accrue only as the new workload is served and indexed---so at
insertion the region is near-empty by construction, and the
$\le 1\%$ ($\lesssim 250$ documents) sparse band is exactly the window the threat
model occupies. While the region is sparse, $p_{\mathrm{cover}} = 1.0$ and
exposure $=$ churn $= 17$, exactly the law.
As legitimate documents accumulate, $p_{\mathrm{cover}}$ collapses
($1.0 \to 0.80 \to 0.29 \to 0.09$ at $5\%, 50\%, 100\%$) and the coupling
dissolves: churn climbs to its cap while exposure \emph{falls} (to 5 at full
density), because the real documents out-compete the adversarial hub in retrieval.

Two readings, both honest, and we state both. First, the boundary confirms the
limitation: the exposure--churn coupling and the clean law are
\emph{emergence-window} results, not universal; beyond a competitor density of a
few percent, fragmentation is a real co-mechanism and the gate is no longer the
governing factor. Second, the boundary is consistent with the worst-case insertion
(\S\ref{sec:model}), which places the hub \emph{at emergence}, when the region is
sparse---so the law applies exactly in the window the threat occupies, and the
dense regime shows the threat self-limits as the domain matures: a populated region
suppresses the hub through ordinary retrieval competition rather than through the
gate. We do not over-claim a universal coupling; we report a measured phase
boundary in competitor density. (On E5 the emergent domain sits only $17^\circ$
off-axis, so $c_0 \ge \tau$ and no window opens---an illumination-safe case
rather than a deficit case, so the dense test is informative on BGE; \S\ref{sec:placement}
explains the E5 geometry.)

\begin{figure}[t]
\centering
\includegraphics[width=\columnwidth]{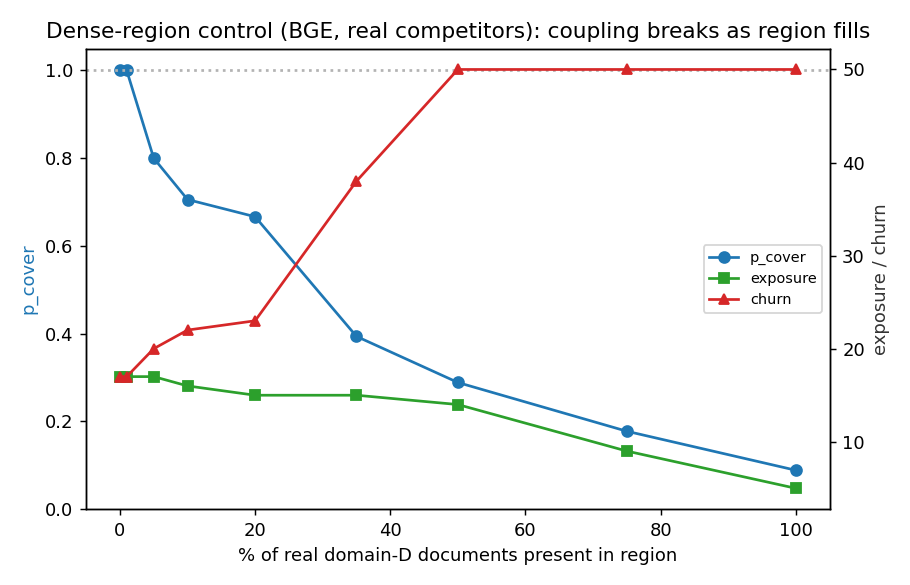}
\Description{As the fraction of real domain-D documents present in the drifted
region rises from 0 to 100 percent, the per-observation cover probability collapses
from 1.0 to 0.09, exposure falls, and churn rises to its cap; the exposure-equals-churn
coupling holds only below about one percent density.}
\caption{Dense-region control on real BGE embeddings with the real in-region
documents as competitors. The exposure--churn coupling (law) holds only in the
emergence regime ($\le 1\%$ density); as the region fills,
$p_{\mathrm{cover}}$ collapses to 0.09 and churn decouples from exposure.
Fragmentation is a real co-mechanism past the boundary.}
\label{fig:dense}
\end{figure}

\subsection{Placement contrast: admission vs.\ retrieval-time normalization}
\label{sec:placement}

The premise that gates everything is one comparison: is admission-time deletion
actually better than a retrieval-time fix that down-weights the hub's score
instead? We run that contrast fairly, on the same real biomedical substrate
(\S\ref{sec:dense}) and one fixed seeded in-region stream replayed against every
arm. The admission count gate is compared against two state-of-the-art
reference-query score normalizers---NNN~\cite{nnn} and QB-Norm~\cite{qbnorm}---each
in two configurations: \emph{static}
(reference bank = established queries, never updated---the condition under which
static normalization is reported to fail against adversarial hubs~\cite{advhub})
and \emph{online} (the bank is updated only at observed in-region query
locations, the same Assumption-2 constraint and budget the sentinel set
operates under; neither defense receives an oracle). Both placements are
calibrated to the \emph{same} benign FPR ($10^{-3}$) that fixes $\tau$, and each
normalizer is given its strongest fair configuration---the largest bias
coefficient whose benign top-1 retention stays within that budget (NNN
$\alpha{=}0.35$ at top-$m{=}30$; QB-Norm $\alpha{=}0.40$ at temperature $0.1$; both
hold benign top-1 retention $\ge 99.9\%$, the same $10^{-3}$ budget as $\tau$, so
the comparator is not hobbled). A defense
\emph{contains} at observation $n$ if the hub is absent from the served top-$k$
for all subsequent in-region queries.

To isolate the mechanism from its placement we run the count test at \emph{both}
placements: the ingest gate (one-shot deletion at admission) and a \emph{serve-time}
count gate that applies the identical $\kappa_S\!\ge\!\tau$ test per query and
suppresses the hub from the served top-$k$ once it trips. The serve-time gate is the
placement-matched comparator for the retrieval-time normalizers---both act at serve
time---so comparing them isolates \emph{count vs.\ normalization} from
\emph{ingest vs.\ serve} (Table~\ref{tab:placement}).

\begin{table}[t]
\centering
\caption{Placement contrast (emergence regime, BGE, D\_treccovid; one fixed
seeded in-region stream). Ingest $=$ one-shot deletion at admission; serve $=$
per-query suppression; normalizers are online, FPR-matched, strongest-fair. The
count test contains at the floor at \emph{both} placements; neither score normalizer
contains at either. The distinguishing axis is the mechanism (count vs.\ normalization),
not the placement (ingest vs.\ serve).}
\label{tab:placement}
\small
\begin{tabular}{@{}llcc@{}}
\toprule
mechanism & placement & exposure & contains? \\
\midrule
count $\kappa_S\!\ge\!\tau$ & ingest & 17 & yes \\
count $\kappa_S\!\ge\!\tau$ & serve  & 17 & yes \\
NNN normalizer              & serve  & 43 & no  \\
QB-Norm normalizer          & serve  & 42 & no  \\
\bottomrule
\end{tabular}
\end{table}

In the emergence regime \emph{both} placements of the count test contain at
$m_{\mathrm{safe}}+\mathrm{deficit}=17$ (ingest and serve give the identical floor,
as they accumulate the identical witnesses), while \emph{neither} normalizer ever
contains: online exposure $43$ (NNN) and $42$ (QB-Norm) across the stream, each
barely improving on its own static arm ($45$, $44$). QB-Norm fails \emph{despite}
a more aggressive fair coefficient ($\alpha{=}0.40$ vs NNN's $0.35$), so the failure
is not insufficient bias but something structural---which is the point: a query-bank
normalizer down-weights the hub \emph{and its in-region neighbours together}---they
share the in-region query mass---so the hub's \emph{rank} is preserved; admission-time
deletion isolates the hub, which is exactly why $\tau$ witnesses suffice. The
advantage is itself an emergence-window result. As in-region documents accrue, the
raw retrieval margin $\mathrm{sim}(h,q)-\mathrm{sim}_{(k)}$ crosses zero (from
$+0.09$ at the emergence instant to $-0.02$ at full density)---the same transition
the per-observation cover probability marks ($1.0\to0.09$, \S\ref{sec:dense}),
measured twice---so beyond that boundary the hub is dropped by retrieval competition
with no defense at all---neither placement governs. The contrast thus confirms the premise without
inverting the honesty of the dense-regime finding: admission is the containing
mechanism precisely in the window the threat occupies, and a fairly-maintained
retrieval-time normalizer is not. We are precise about what separates, and we
\emph{measure} it rather than assert it: the same count test $\kappa_S\!\ge\!\tau$
placed at serve time contains at exactly $17$ (Table~\ref{tab:placement}), identical
to the ingest gate, because it accumulates the very same witnesses---so the count
mechanism is placement-invariant, whereas score normalization contains at neither
placement. The distinguishing axis is therefore \emph{count versus normalization}, not
ingest-versus-serve; admission is the ingest-path placement of the count
mechanism, and the maintenance cost of \S\ref{sec:bound}--\S\ref{sec:coupling} is
the cost of \emph{that} placement. The serve-time arm's distinction is thus not its
exposure but its cost accounting, and it maps exactly onto the
$E=E_{\mathrm{obs}}+E_{\mathrm{rev}}$ decomposition of \S\ref{sec:model}: both
placements pay the identical observation-limited floor $E_{\mathrm{obs}}=17$ (the
same witnesses), and they differ only in $E_{\mathrm{rev}}$, the re-evaluation
overhead. The ingest gate evaluates the count once at admission and pays
$E_{\mathrm{rev}}$ as a coarse revocation cadence; the serve-time gate is the
maximum-cadence limit of that same term, re-evaluating the count on every query.
The serve-time arm therefore confirms, rather than complicates, the decomposition:
it is $E_{\mathrm{rev}}$ at full cadence, leaving the floor untouched.

\paragraph{Generalization across emergent domains.} Rotating which BEIR domain is
the emergent off-axis region (the other three established) at the fixed operating
point, all four open a window on BGE (deficit $11$--$13$). This is a genuine
\emph{cross-prediction}, not a back-solve. For each domain we measure
$p_{\mathrm{cover}}$ \emph{directly}---the fraction of in-region queries that
retrieve the hub $h$ in the top-$k$, the detection-phase cover quantity, computed
independently of any exposure count---then form the predicted floor
$p_{\mathrm{cover}}\,m_{\mathrm{safe}}+\mathrm{deficit}$ and compare it against the
floor \emph{measured by running the gate} on $20$ randomized arrival orders
(Table~\ref{tab:crosspred}). The independently measured $p_{\mathrm{cover}}$
predicts the observed floor to within $0.5$ in every domain; no parameter is fit to
the exposure.

\begin{table*}[t]
\centering
\caption{Four-domain cross-prediction (BGE). $p_{\mathrm{cover}}$ is measured
\emph{directly} as the in-region hub-retrieval rate, independent of any exposure
count; the predicted floor $p_{\mathrm{cover}}\,m_{\mathrm{safe}}+\mathrm{deficit}$
($m_{\mathrm{safe}}{=}5$) is then compared against the floor measured by running the
gate over $20$ randomized arrival orders. Measured $p_{\mathrm{cover}}$ predicts the
observed floor to within $0.5$ throughout---the prediction is genuine, not fit.}
\label{tab:crosspred}
\small
\begin{tabular}{lcccc}
\toprule
emergent domain & deficit & measured $p_{\mathrm{cover}}$ & predicted floor & observed floor \\
\midrule
A\_naturalquestions & 11 & 0.84 & 15.2 & $14.9\pm1.0$ \\
B\_fiqa             & 13 & 0.91 & 17.5 & $17.2\pm0.7$ \\
C\_scidocs          & 13 & 0.78 & 16.9 & $16.4\pm1.0$ \\
D\_treccovid        & 12 & 1.00 & 17.0 & $17.0\pm0.0$ \\
\bottomrule
\end{tabular}
\end{table*}

Where Lemma~\ref{lem:cover}'s hypothesis holds the prediction is exact:
D\_treccovid has measured $p_{\mathrm{cover}}=1.00$, predicting
$m_{\mathrm{safe}}+\mathrm{deficit}=17$, and the gate floor is $17.0\pm0.0$. Where
real $p_{\mathrm{cover}}<1$---a detection-phase query covers the region without
exposing $h$---the floor bends below $m_{\mathrm{safe}}+\mathrm{deficit}$ by exactly
$p_{\mathrm{cover}}\,m_{\mathrm{safe}}$, as the conditional bound predicts: e.g.\
A\_naturalquestions, measured $0.84\!\to\!$ predicts $15.2\!\to\!$ observed $14.9$.
The data thus confirm the conditional bound and its $p_{\mathrm{cover}}$ dependence,
not a violation of a bound stated without it. Across all four, \emph{neither} online
normalizer contains (NNN $45$--$73$, QB-Norm $41$--$62$ exposed). The real-embedding
evidence is thus a four-domain generalization with randomized arrival that
cross-predicts the floor from an independently measured $p_{\mathrm{cover}}$, not a
single trajectory.

\paragraph{Why E5 opens no window: the illumination-safe regime, measured.} On E5
all four domains are instead illumination-safe ($c_0\ge\tau$ at $\tau{=}18$)---no
window opens---so the real-geometry results are on BGE. This is not an idiosyncrasy
of E5 but an instance of the illumination-safe corner the phase structure predicts
(\S\ref{sec:phase}), and the geometry says why. E5 embeds the four BEIR query
domains at roughly \emph{half} the off-axis angle to the established-workload
centroid that BGE does---mean $18^\circ$ (per-domain $17$--$20^\circ$) versus BGE's
$36^\circ$ ($34$--$39^\circ$)---so each emergent region falls \emph{inside} the cone
the established sentinels' top-$k$ already illuminates, giving $c_0\ge\tau$ before
any drift. The cause is specifically this angular placement, not a globally
lower-dimensional space: the participation-ratio effective dimensionality of the two
query spaces is comparable ($87$ for E5, $96$ for BGE), so E5 is not collapsing the
geometry---it is placing these particular domains nearer the established axis. E5
therefore does not weaken the BGE result; it exhibits the no-deficit corner of the
same phase structure, where a region is so well-covered at insertion that no window
ever opens.

\subsection{Steady-state recurrence}
\label{sec:steady}
The floor is not a one-off transient. We sub-cluster the three drifted BEIR
domains into $24$ real emergent regions (k-means) against the established
workload: \emph{all} $24$ open a window ($\mathrm{deficit}>0$), with deficits
around a median of $11$ (min $1$, max $14$)---the missing-coverage mass is a
population property, not a property of one hand-picked domain. In a continuously
drifting workload each emerging region pays its own
$m_{\mathrm{safe}}+\mathrm{deficit}$ in ingest churn, so cumulative churn grows
linearly with the number of regions spawned ($19,69,127,362$ after $1,5,10,24$
regions; averaging $\approx 15$ promotions/region, higher for the early regions---the
first costs $19$, as larger deficits tend to surface first). At a fixed observation horizon the
ingest-path cost therefore scales with the workload's region-spawn rate: the floor
is paid \emph{per emerging region}, not once, which is the sense in which the
emergence window is a steady-state stream rather than a transient.

\input{section_ann_systems}

\input{section_experiments}

\section{The Operational Verdict}
\label{sec:verdict}

\paragraph{Budget-independent and established (in the emergence regime).} (i) in
the emergence window, the mechanism is the witness deficit, not
fragmentation; (ii) it obeys
$E = m_{\mathrm{safe}}+\lceil(\tau -c_0)_+/r\rceil$; (iii)
stopping at coverage is strictly and largely worse than promoting to redundancy; (iv) in that regime exposure and churn are governed by the same deficit
term (\S\ref{sec:coupling}), giving budget-parametric infeasibility; (v) all of the
above replicate across embedding families---the \emph{synthetic} mechanism
gates; the real-embedding results of (vi) and \S\ref{sec:dense}--\ref{sec:steady} are
BGE only, as E5 opens no window at these operating points; and (vi) admission-time
deletion contains the hub in this window where a fairly-maintained, FPR-matched
retrieval-time normalizer does not (\S\ref{sec:placement}). The scope is the emergence
window: \S\ref{sec:dense} locates the competitor-density boundary past which the
coupling dissolves and retrieval competition, not the gate, governs the hub.

\paragraph{Budget-dependent, by construction.} Whether the contained exposure (19 on
BGE, 23 on E5) and the churn it costs cross a given operational budget depends on the
deployment exposure budget $B_{\mathrm{exp}}$, the churn budget $B_{\mathrm{churn}}$,
and the calibrated $\tau$. On the present reconstruction the exposure floor sits
below the smallest tested $B_{\mathrm{exp}}$, while the triad binds on
$B_{\mathrm{churn}}$; E5's larger $\tau$ pushes its floor (23) close to
$B_{\mathrm{exp}} = 30$. These are the deployment quantities; the deployment verdict is
therefore a calibrated-workload question rather than a single label. It is the function
$\mathrm{verdict}(\tau, B_{\mathrm{exp}}, B_{\mathrm{churn}})$, and pinning it
requires the real workload budgets and the locked $\tau$.

\paragraph{Why the emergence window matters.} The window self-heals once legitimate
documents arrive, but it still matters for two reasons. First, \emph{emergence windows
recur}. In a continuously drifting
workload---new query regions opening, the corpus growing (\S\ref{sec:model})---the
window is not a one-off transient but a \emph{steady-state stream}: each newly
emerging region opens a fresh deficit, so the floor (and its coupled ingest churn)
is paid repeatedly, at a rate set by how fast the workload spawns new regions, not
once. We measure this directly (\S\ref{sec:steady}): across $24$ real sub-regions
all open a window and cumulative ingest churn grows linearly in the number
spawned, so at fixed time the cost scales with the spawn rate. Second, the dense regime is itself a \emph{wasted-cost} result. Past the
boundary of \S\ref{sec:dense}, churn climbs to its cap while exposure falls to $5$:
the monitor spends escalating ingest-path cost chasing containment of a hub that
retrieval competition has already neutralized. That is maintenance cost
decoupling from---and then exceeding---its protective value, exactly the kind of
ingest-path inefficiency an admission-index designer must budget for. The
emergence regime is where the gate earns its keep; the dense regime is where it
keeps spending after it has stopped helping.

\paragraph{The open axes.} It is worth being explicit about where this work is
closed and where it is open. On the \emph{mechanism} axis, the emergence window is
closed: in the sparse regime the fragmentation control rules it out and identifies the
deficit, and the law is exact and cross-family. That same axis is \emph{open in the
dense regime}: past the competitor-density boundary of \S\ref{sec:dense},
fragmentation co-governs and the gate is no longer the controlling factor---a
boundary we measure but do not push further (a closed-loop adversary that exploits
it is future work). On the \emph{regime/budget} axis, the question is open by
construction: whether the established floor is operationally significant is
undecided until $\tau$ and the budgets are deployment-calibrated. We regard this as
the honest state of the result rather than a gap to paper over---the emergence-window
mechanism is settled, its dense-regime boundary is measured, and the operational
consequence is the live question we resist resolving prematurely with reconstructed
quantities.

\paragraph{A design point the placement contrast names.} The contrast is also a
constructive critique of the defense literature, not only a negative result about
our gate's competitors. Deployed serve-time hubness defenses choose
\emph{normalization}---they down-weight a hub's score---when, by
\S\ref{sec:placement}, a serve-time \emph{count} test ($\kappa_S\ge\tau$ applied at
retrieval) \emph{does} contain the very hubs normalization cannot---measured at
exposure $17$ in Table~\ref{tab:placement}---because a count isolates the hub where a
shared-query-bank bias cannot. The unexplored design point is thus a serve-time count
gate; it pays the identical observation-limited floor $E_{\mathrm{obs}}$, spending its
re-evaluation overhead $E_{\mathrm{rev}}$ as per-query suppression rather than as
ingest-time churn. Choosing between the two placements of the count is therefore a
cost-allocation decision over $E_{\mathrm{rev}}$, while choosing normalization over
the count at either placement is, in the emergence window, simply the wrong
mechanism.

\section{Related Work}

\paragraph{Database systems and incremental maintenance.} The admission check is a
query-aware auxiliary index, and our cost lens is the database one. Vector stores
serve top-$k$ retrieval over ANN indexes---HNSW~\cite{hnsw}, product quantization
and inverted files~\cite{pq}, and GPU- or disk-resident
variants~\cite{faiss,diskann,scann}---in purpose-built systems~\cite{milvus}. A
growing line keeps these indexes correct \emph{online}, and it is the work closest to
ours in spirit. FreshDiskANN~\cite{freshdiskann} supports real-time inserts and
deletes on a disk-resident graph index, preserving recall through a streaming merge
of an in-memory delta into the on-disk graph; SPFresh~\cite{spfresh} applies in-place,
locally rebalanced updates to inverted posting lists so a partitioned index stays
fresh without global rebuilds; DeDrift~\cite{dedrift} targets \emph{distribution
drift} directly, periodically re-fitting IVF centroids as the indexed data shifts
beneath them; and adaptive IVF maintenance~\cite{adaivf} re-tunes partitioning under
streaming updates. The sentinel set is likewise an incrementally maintained view, so
classical incremental view maintenance~\cite{ivmsurvey} with its
higher-order~\cite{dbtoaster} and dataflow~\cite{diffdataflow} descendants, adaptive
self-tuning indexing~\cite{cracking}, stream drift detection~\cite{adwin}, and
incremental reverse-$k$NN~\cite{rknn} are its closest analogues. Every one of these
systems, however, is \emph{handed its deltas}: FreshDiskANN and SPFresh are given the
inserted and deleted vectors, DeDrift is given the drifted data whose centroids it
re-fits, and IVM is given the base-table updates whose effect it propagates. Their
problem is to apply a \emph{known} delta at low latency and cost while preserving
recall. The constraint that makes our floor non-trivial sits upstream of all of them:
the update that closes the gate---a sentinel covering the emerging region---cannot be
\emph{computed} from the arrived document, because it does not exist until an
in-region query is \emph{observed}. The maintenance is therefore
observation-limited, and even a maintenance engine with zero update latency and zero
rebuild cost cannot pre-empt the exposure floor, because the floor counts how many
query observations must accrue, not how fast each is applied
(Corollary~\ref{cor:latency}). That gap---between an observation-limited acquisition
bound and the delta-driven maintenance machinery---is our contribution, not the
machinery itself; Table~\ref{tab:related} summarizes the contrast.

\begin{table}[t]
\centering
\caption{This work versus online index-maintenance systems. Prior systems apply a
\emph{known} delta at low cost while preserving recall; none models the
observation-limited acquisition cost of an ingest-time admission index under query
drift, nor the exposure it leaves open.}
\label{tab:related}
\small
\begin{tabular}{@{}lcccc@{}}
\toprule
system & online & query & admission & cost \\
       & maint. & drift & index     & model \\
\midrule
FreshDiskANN~\cite{freshdiskann} & \checkmark & -- & -- & -- \\
SPFresh~\cite{spfresh}           & \checkmark & -- & -- & -- \\
DeDrift~\cite{dedrift}           & \checkmark & \checkmark & -- & -- \\
adaptive IVF~\cite{adaivf}       & \checkmark & partial & -- & -- \\
IVM~\cite{ivmsurvey}             & \checkmark & -- & -- & delta \\
\midrule
\textbf{this work} & \checkmark & \checkmark & \checkmark & \textbf{obs.-limited} \\
\bottomrule
\end{tabular}
\end{table}

\paragraph{Hubness and reference-query normalization.} Hubness in high-dimensional
nearest-neighbour search is well studied~\cite{hubness}, with a long line of
reduction methods---mutual proximity and local scaling~\cite{mutualproximity},
surveyed comprehensively in~\cite{hubsurvey}. The reference-query normalization
family---inverted softmax~\cite{invsoftmax}, NNN~\cite{nnn}, QB-Norm~\cite{qbnorm},
DBNorm~\cite{dbnorm}---and training-time balancing (NeighborRetr~\cite{neighborretr})
attack hubness on the \emph{retrieval-accuracy} axis, and they share an implicit
\emph{placement} commitment: hubness is corrected by re-scoring at query time, so a
hub is suppressed by subtracting a popularity bias estimated against a reference-query
bank. We share the maintained-reference-query structure but make two departures.
First, our axis is enforcement under an adversary with an exposure clock, not
retrieval quality, and the contribution is the coverage--redundancy bound. Second,
and more pointedly for this family, \S\ref{sec:placement} shows that the placement
commitment is the wrong one for adversarial \emph{defense}: running the count test at
both the ingest and serve placements and the strongest fair NNN and QB-Norm
normalizers at the serve placement, the count test contains the hub at either
placement (exposure $17$) while normalization contains at neither ($42$--$43$). The
reason is structural---a query-bank normalizer down-weights the hub \emph{and} its
in-region neighbours together, because they share the in-region query mass, so the
hub's \emph{rank} is preserved no matter how aggressively the bias is tuned within a
fair benign-FPR budget. The distinguishing axis for hub containment is therefore the
mechanism---a count of distinct witnessing queries---not the placement at which a
score is normalized. This reframes retrieval-time normalization as the wrong tool for
adversarial hub containment, independent of how well its reference bank is
maintained.

\paragraph{Adaptive defenses under attack.} That adaptive filters are exploitable by
dynamic attackers is established outside our setting (online-filter poisoning;
Stackelberg-game robust learning; attacks on test-time
adaptation)~\cite{onlinepoison,stackelberg}. We contribute a geometry-specific,
latency-independent lower bound and a churn coupling, rather than a generic
exploitability phenomenon.

\paragraph{Adversarial hubs and corpus poisoning.} The hub threat itself is prior
art---adversarial multi-modal hubs~\cite{advhub}, knowledge-base and corpus
poisoning of retrieval-augmented systems~\cite{poisonedrag,corpuspoison}; we treat
it as the premise, not a contribution. Two recent papers study the
security reading of this same sentinel admission gate directly---global gating
against coordinated poisoning~\cite{pathakGating} and the containment limits of the
sentinel count under poisoning~\cite{pathakContainment}. The present paper is
distinct and complementary: it takes the gate as given and quantifies
the \emph{ingest-path maintenance cost} of keeping it sound as the query workload
drifts---a data-management question---rather than adversarial containment
guarantees. The exposure lower bound, the approximate-index failure mode, and the
maintenance-churn coupling studied here do not appear in that security work.

\section{Limitations and Threats to Validity}
\label{sec:limits}

Calibration is on a reconstruction; $\tau$ is not locked to a separate stationary
study, and $B_{\mathrm{exp}}, B_{\mathrm{churn}}$ are placeholders---hence the
deliberate refusal to issue an operational verdict. Drift uses a synthetic rotation
control for a clean $\theta$ sweep; the real-embedding phase structure is
corroborated but the magnitude-vs-budget question awaits deployment calibration.
Fragmentation is falsified as causal within the sparse regime (off-axis regions
empty of competitors); the dense boundary, where competitors are in-region, is now
tested directly (\S\ref{sec:dense}) rather than deferred, and there fragmentation is
a real co-mechanism---so the coupling and the clean law are scoped to the emergence
window, not claimed universally. The triad search
is finite, but Corollary~\ref{cor:triad} supplies the structural reason the frontier
holds beyond the grid. A closed-loop adversary that parks in the under-covered zone,
shifts region after each promotion, or weaponizes a redundancy-aware monitor's churn
budget is left to future work; the present bound is generous to the defender (it
assumes continued promotion), so the realistic coupled interaction can only enlarge
the floor.

\section{Conclusion}

Keeping a query-aware admission index sound under workload drift imposes an
ingest-path maintenance cost that online index-maintenance systems do not model. The
monitor's stopping condition---coverage by one sentinel---is weaker than the
predicate's rejection condition of $\tau$ witnesses, and closing that gap is
observation-limited: each witness requires a fresh in-region query, so the exposure
floor $m_{\mathrm{safe}}+\lceil(\tau-c_0)_+/r\rceil$ is independent of update and
enforcement latency, and the same promotions that close the gap are the events that
expose the region. On real ANN indexes the floor is a best case---approximate
retrieval inflates it, and below a recall threshold the gate fails to contain until the
monitor counts witnesses on a higher-recall probe. The check itself is a fixed
$O(|S|d)$ cost, negligible beside the insert it rides on; the cost that scales is
exposure, and the bound that predicts it also provisions against it. Coverage is not
redundancy: locally sufficient maintenance decisions compose into a globally
unaccounted ingest-path cost wherever a coverage-driven maintenance layer feeds a
threshold predicate.

\appendix
\section{Additional Synthetic Controls}

\subsection{Distribution and budget sensitivity}

Two claims of different kinds. \emph{Budget-independent:} the effect is a population
property, not a few pathological hubs---the top 5\% of hubs hold only 8\% of the
total deficit (median deficit $d = 14$ over this per-hub population, distinct from
the $24$ coarser spawned sub-regions of \S\ref{sec:steady}, whose median is $11$).
\emph{Budget-dependent:} sweeping the
placeholder low budget $\pm 50\%$, the fixed floor (an independent quantity) crosses
that budget only at $0.63\times$ ($-37\%$) and is stable to $\pm 30\%$---a smooth
single-threshold crossing, not a knife-edge flip. The crossing is a statement about
where a chosen budget falls relative to the floor, not about the floor itself; the
budget here is a placeholder, so this characterizes sensitivity, not a verdict.

\subsection{Budget feasibility}

We separate an empirical claim from a structural one, because they have different
scope. \emph{Empirically (regime-scoped):} on the present calibration no
configuration meets exposure $\le 30$, FPR $\le 10^{-2}$, and churn $\le 5$
simultaneously; the binding constraint is churn---continued promotion meets exposure
(19) and FPR but costs 19 witnesses ($> 5$), while novelty-stopping holds churn at 1
but never contains, so exposure grows without bound (248 at the run horizon). This is
a statement about the observed/calibrated
regime, not a universal one, and we state it as such. \emph{Structurally
(budget-parametric):} Corollary~\ref{cor:triad} is stronger and does not depend on
the particular budgets. Because the witnesses that close the gate are the
promotions, the deficit lower-bounds churn-to-close; hence for any
$B_{\mathrm{churn}} < m_{\mathrm{safe}} + \lceil(\tau - c_0)_+/r\rceil$, no
$B_{\mathrm{exp}}$ admits a containing novelty-only configuration. The empirical ``0
feasible configs'' is one point on this frontier; the frontier itself is the general
result.

\subsection{Adaptive versus static maintenance}

Against a never-adapting baseline, adaptation shifts the exposure/FPR frontier
inward by $\sim 232$ at matched FPR. We frame this as a frontier shift at a churn
cost, not as domination: at low $\theta$ static is already optimal
($c_0 \ge \tau$, floor 0), and the high-$\theta$ gain is the same
continued-vs-novelty quantity ($248 - 19 \approx 232$) viewed against no adaptation.
Adaptation moves the operating point along the exposure/churn frontier; it does not
provide a free lunch.

\section*{Artifacts}
All code and result data for this paper are publicly available at
\url{https://github.com/p23pathak/query_staleness_edbt_2027}. The repository contains the
experiment harness---the admission-gate and recall-aware mitigation runs on real
HNSW/IVF/IVF-PQ indexes, the Wikipedia-clickstream drift replay, the
PostgreSQL/\texttt{pgvector} prototype, and the ingest-cost, scaling, provisioning, and
cross-encoder measurements---together with the result file behind every table and figure,
so the reported numbers can be inspected directly without re-running. A \textsc{README}
documents the software environment and the command-to-result mapping. The public
MS~MARCO, TREC-COVID, and Wikipedia-clickstream corpora are not redistributed; the
repository documents how to download them and regenerate the \texttt{bge-large}
embeddings. No credentials or private data are required to run the analysis.

%% file: section_ann_systems.tex
\section{The Floor on Real Vector Indexes}
\label{sec:ann}

Every result so far computes the admission predicate with an \emph{exact}
reverse-$k$NN: $\kappa_S(d)$ is evaluated by brute-force top-$k$ over the corpus.
A production vector store does not serve exact top-$k$; it serves an
\emph{approximate} nearest-neighbour index (HNSW, IVF, or a quantized IVF-PQ) whose
recall is below one and is itself a tuning knob. Because both the served retrieval
\emph{and} the sentinel witness count are then read off that same approximate index,
the observation-limited floor of \S\ref{sec:bound} is a statement about the exact
retriever, and the operative question for a data-management system is whether it
survives approximate retrieval. We answer it directly on a real index at production
scale.

\paragraph{Setup.} We build the established store from the MS MARCO passage corpus
($8{,}841{,}823$ passages) encoded with \texttt{bge-large-en-v1.5} (1024-d,
unit-normalized), with the biomedical TREC-COVID domain as the emergent off-axis
region (hub $=$ its centroid; established sentinels drawn from the MS MARCO query
workload). Over the $\sim$8.8M-vector pool we build three indexes spanning the
deployment spectrum---HNSW (graph), IVF-Flat (inverted list), and IVF-PQ (the
memory-compressed index used at billion scale)---and sweep each one's recall knob
(HNSW \texttt{efSearch}, IVF \texttt{nprobe}). For every configuration we measure
recall@$k$ against exact GPU ground truth and then run the \emph{identical}
continued-promotion containment simulation of \S\ref{sec:results} through that index,
recording the churn to contain, the \emph{served} exposure (in-region queries for
which the index returns the hub), and the \emph{true} exposure (queries that retrieve
the hub under exact ground truth---the real risk, part of which a low-recall index
silently drops). Operating point $\tau=14$, $m_{\mathrm{safe}}=5$, $k=10$; the hub
sits at $c_0=0$, deficit $14$, so the analytical floor is
$m_{\mathrm{safe}}+\mathrm{deficit}=19$.

\begin{table}[t]
\centering
\caption{The witness-deficit floor on real ANN indexes over 8.8M MS MARCO vectors
(BGE-large; emergent TREC-COVID; $\tau{=}14$, deficit $14$, analytical floor $19$).
As index recall falls, churn-to-contain and \emph{true} exposure rise above the floor,
and below recall $\approx0.5$ the monitor fails to contain at all. IVF-PQ is
recall-capped by quantization.}
\label{tab:ann}
\small
\begin{tabular}{@{}llccccc@{}}
\toprule
index & knob & recall & served & true & churn & contains \\
\midrule
exact       & ---            & 1.000 & 17.5 & 17.5 & 24.5 & 1.00 \\
\midrule
HNSW        & efS$=256$      & 0.966 & 17.5 & 17.5 & 24.5 & 1.00 \\
HNSW        & efS$=128$      & 0.912 & 17.5 & 19.0 & 27.0 & 1.00 \\
HNSW        & efS$=64$       & 0.838 & 16.5 & 21.0 & 31.5 & 1.00 \\
HNSW        & efS$=16$       & 0.636 & 16.5 & 31.5 & 45.5 & 1.00 \\
HNSW        & efS$=8$        & 0.474 & 14.0 & 35.0 & 50.0 & \textbf{0.25} \\
\midrule
IVF-Flat    & nprobe$=128$   & 0.900 & 17.5 & 17.5 & 24.5 & 1.00 \\
IVF-Flat    & nprobe$=32$    & 0.742 & 17.0 & 19.0 & 27.0 & 1.00 \\
IVF-Flat    & nprobe$=8$     & 0.554 & 15.5 & 26.5 & 36.0 & 1.00 \\
IVF-Flat    & nprobe$=4$     & 0.456 &  8.0 & 35.0 & 50.0 & \textbf{0.00} \\
IVF-Flat    & nprobe$=1$     & 0.258 &  2.0 & 35.0 & 50.0 & \textbf{0.00} \\
\midrule
IVF-PQ      & nprobe$=64$    & 0.308 & 16.0 & 23.5 & 34.5 & 1.00 \\
IVF-PQ      & nprobe$=8$     & 0.280 & 15.0 & 29.0 & 41.0 & 1.00 \\
\bottomrule
\end{tabular}
\end{table}

\paragraph{The exact floor is a best case (real competitors already inflate it).}
Even at recall $1.0$ the measured churn is $24.5$, above the $19$ floor, because at
8.8M the real corpus contains near-region competitors: the per-observation cover
probability is $p_{\mathrm{cover}}=0.72<1$, so containment needs
$\mathrm{deficit}/p_{\mathrm{cover}}\approx19.4$ promotions rather than $14$. This is
the dense-regime effect of \S\ref{sec:dense} arising \emph{naturally} from a real
background rather than from injected competitors, and it makes
$m_{\mathrm{safe}}+\mathrm{deficit}$ a genuine lower bound that a real store sits
above.

\paragraph{Approximation inflates the floor monotonically.} Reducing recall raises
both churn-to-contain and true exposure monotonically
(Table~\ref{tab:ann}, Fig.~\ref{fig:ann}): at a usable HNSW recall of $0.64$ the
monitor still contains but pays churn $45.5$ ($1.9\times$ the exact $24.5$) and leaves
true exposure $31.5$ ($1.8\times$). The mechanism is that an approximate sentinel
retrieval covers the hub less often ($p_{\mathrm{cover}}$ falls with recall), so each
promotion is a weaker witness and more promotions---hence more in-region
arrivals---are needed. The floor is thus not merely a function of the deficit; it is a
function of the deficit \emph{and} the index recall.

\paragraph{A containment-failure threshold invisible to the exact analysis.} Below
recall $\approx0.5$ the monitor can stop containing the hub \emph{entirely}: HNSW at
\texttt{efSearch}$=8$ contains in only $25\%$ of arrival orders, and IVF-Flat at
\texttt{nprobe}$\le4$ \emph{never} contains within the stream. The extreme
(\texttt{nprobe}$=1$, recall $0.26$) is the starkest case: the index returns
the hub for only $2$ in-region queries---so the monitor accumulates almost no
witnesses---while exact retrieval returns it for $35$, and containment never fires.
A low-recall index therefore does not just cost more; it can convert a contained
threat into an uncontained one that the admission monitor cannot even see, a failure
mode the exact-retriever analysis of \S\ref{sec:bound}--\S\ref{sec:coupling} cannot
express.

\paragraph{The compressed index is the worst case.} IVF-PQ---the index actually
deployed when the corpus does not fit in RAM---is recall-capped by quantization at
$\approx0.31$ regardless of \texttt{nprobe}, so it pays a \emph{permanent}
$\sim\!1.5\times$ churn penalty (34.5 vs 24.5) and $\sim\!1.3\times$ true exposure.
Unlike low-\texttt{nprobe} IVF-Flat it still contains, because its recall loss is
spread evenly across queries (keeping $p_{\mathrm{cover}}\approx0.45$) rather than
concentrated. The ordering across index families---graph $\succ$ inverted-list
$\succ$ quantized---is exactly the memory/recall trade-off a systems designer picks
along, now priced in admission-floor terms.

\begin{figure}[t]
\centering
\includegraphics[width=\columnwidth]{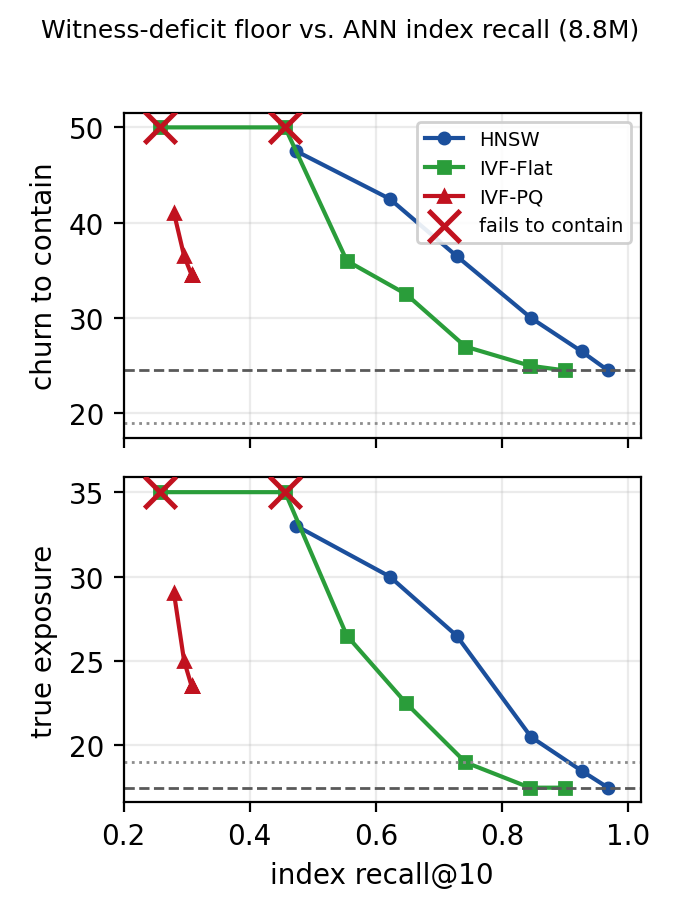}
\Description{Churn-to-contain and true exposure rise as index recall falls, across
HNSW, IVF-Flat, and IVF-PQ, with a shaded region below recall 0.5 where the monitor
fails to contain.}
\caption{The floor on real 8.8M-vector indexes. As recall@10 falls (right to left),
churn-to-contain (top) and true exposure (bottom) rise above the exact values
(dashed) and the analytical floor $19$ (dotted). $\times$-marked configurations fail
to contain the hub within the stream (HNSW \texttt{efSearch}$=8$; IVF-Flat
\texttt{nprobe}$\le4$). IVF-PQ (recall $\approx0.31$, quantization cap) still contains
despite low recall---its recall loss is spread across queries rather than concentrated
on a few---so containment failure is not a clean recall threshold.}
\label{fig:ann}
\end{figure}

\subsection{The ingest-path cost, measured}
\label{sec:ann_ingest}

The maintenance argument of \S\ref{sec:bound} claims a cost on the ingest path; we
measure it on the real store. Insertion is the CPU HNSW \texttt{add} (as in every
production HNSW deployment); the admission check is the reverse-$k$NN count
$\kappa_S(d)$, computed as one $(|S|\times d)$ matrix--vector product against the
sentinel thresholds followed by a compare---no per-sentinel re-search.

\paragraph{Admission is a fixed, corpus-independent cost, dominated at scale by the
insert.} Over the 8.8M store, a single HNSW insert costs $\approx47$\,ms (median
$46.5$, p99 $51$; $\approx21$ docs/s, warmup discarded), while the $\kappa_S$ check
costs $0.04$\,ms at $|S|=5000$ sentinels---a relative ingest tax of $0.08\%$.
Crucially the check is $O(|S|\cdot d)$ and
\emph{independent of corpus size}: it scales linearly in the sentinel budget
($0.02$\,ms at $|S|{=}1000$ to $0.22$\,ms at $|S|{=}20000$) but not in $N$, whereas
the ANN insert grows with $N$. This reconciles scale regimes---at a $100$k store,
where an insert is $\approx1$\,ms, the same check is a $\sim\!37\%$ tax, but as the
store grows to production scale the insert dominates and the admission overhead
becomes negligible. The maintenance cost that matters at scale is therefore
\emph{not} the CPU cost of the check---that is nearly free---but the \emph{exposure
and churn} the approximate index imposes (\S\ref{sec:ann}). Revocation re-evaluation
scales linearly with cadence ($\approx4.9$/$13.9$/$35$\,ms per $100$/$500$/$2000$
admitted documents at $|S|=5000$), so a coarser cadence trades $E_{\mathrm{rev}}$
against staleness exactly as the $E=E_{\mathrm{obs}}+E_{\mathrm{rev}}$ decomposition
of \S\ref{sec:model} predicts.

\subsection{Restoring Containment: A Recall-Aware Fix}
\label{sec:mitigation}

The silent blind spot is a deployment bug with a cheap fix, and identifying it points
straight at the correction. Its cause is that the monitor counts witnesses through the
same low-recall serving index that misses the hub; the two roles---serving traffic and
counting admission witnesses---need not share a recall. Two corrections decouple them
(Fig.~\ref{fig:ann_mitigation}). \emph{fixA} lowers the trip threshold in proportion to
the measured recall, $\tau_{\mathrm{eff}}=\lceil\tau\cdot\mathrm{recall}\rceil$,
compensating the approximate under-count. \emph{fixB} counts witnesses on a
\emph{high-recall probe} decoupled from the serving index---exact for the $|S|$
sentinels only, which is affordable precisely because the monitor probes a fixed
sentinel set rather than the whole corpus, while serving stays approximate for speed.

On the 8.8M store, \emph{fixB restores containment to $1.0$ and pins true exposure to
the exact floor $17.5$ at every recall operating point}---including the three IVF-Flat
configurations ($\mathrm{nprobe}\le4$) where the baseline never contains and the HNSW
$\mathrm{efSearch}{=}8$ point where it contains only half the time. fixA is a cheaper
partial fix: it restores containment except at the most degenerate recall
($\mathrm{nprobe}\le2$, recall $\le0.36$), where a recall-corrected threshold cannot
overcome the concentrated recall loss. The design rule is concrete---an operator running
an approximate index should count admission witnesses on a higher-recall probe than the
one serving traffic; the cost is exact scoring of a fixed sentinel budget, not of the
corpus.

\begin{figure}[t]
\centering
\includegraphics[width=0.78\columnwidth]{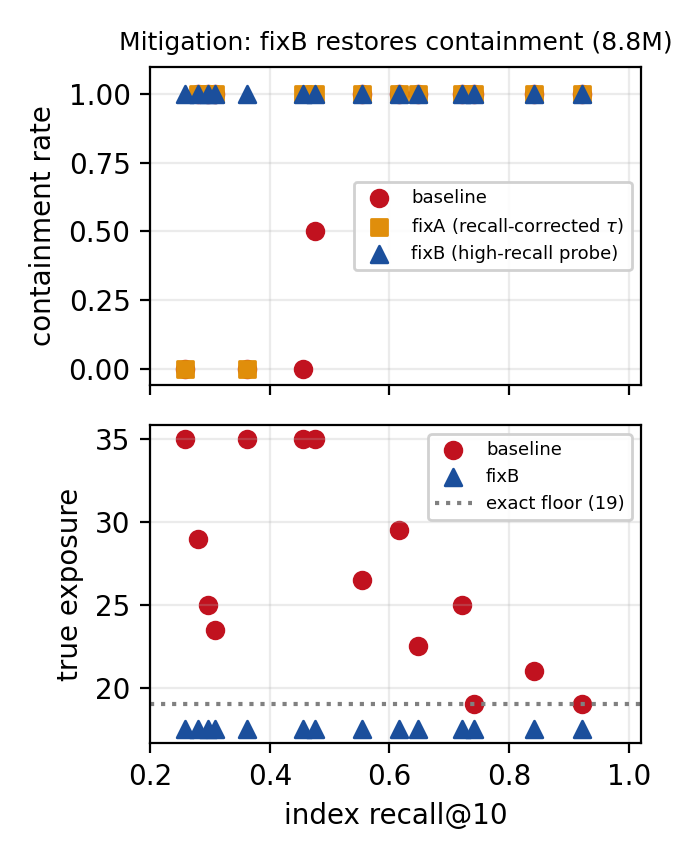}
\Description{Baseline containment collapses at low recall while fixB restores it to one and pins true exposure to the floor across all operating points.}
\caption{Restoring containment (8.8M). Top: containment rate vs.\ index recall---the
baseline (monitor on the serving index) collapses at low recall, fixB (high-recall
probe) holds at $1.0$ everywhere, fixA is partial. Bottom: fixB pins true exposure to
the exact floor ($17.5$) across all operating points where the baseline leaks up to
$35$.}
\label{fig:ann_mitigation}
\end{figure}

%% file: section_experiments.tex
\section{Real Drift, a Prototype, and Provisioning}
\label{sec:experiments}

We close by answering three questions a deployment audience asks of the analysis: does
the law hold under \emph{real} drift rather than a synthetic rotation; what does the
mechanism cost inside an \emph{actual} vector database; and how does one \emph{use} the
bound to configure a system. We also confirm the findings are not embedding-specific.

\subsection{The Law Under Real Temporal Drift}
\label{sec:drift}

The controlled rotation of \S\ref{sec:phase} sweeps the deficit cleanly but is
synthetic. To test the law under organic drift we replay the Wikipedia
\emph{clickstream}---monthly aggregated reader navigation---across the COVID-19
emergence: the established workload is the November--December 2019 top query pages and
the drift period is March--April 2020, when a genuinely novel topic cluster (the
pandemic) appears that the pre-2020 workload cannot cover. Clustering the queries and
taking the established-period pages as sentinels, five regions emerge carrying a
deficit against those sentinels (deficits $2$--$11$, median $5$)---real under-covered
emerging topics, not a constructed angle. (Six \emph{consecutive} 2023 months yield
\emph{no} emerging region---popular topics are semantically stable month to month---so
the shift must be genuine to open a window.)

This is the test the synthetic sweep cannot provide: the floor
$m_{\mathrm{safe}}+\mathrm{deficit}$ is predicted per region from each region's measured
deficit, then checked against the exposure actually incurred when the five regions are
replayed in their true monthly arrival order (Fig.~\ref{fig:drift_law})---an
out-of-sample prediction that could have deviated but does not. Measured exposure matches
the predicted floor region by region (slope $1.0$, $R^2=1.0$; exposure equal to churn,
median difference $0$). The law and the cost--exposure coupling therefore hold unchanged
under a real distribution shift, not only the synthetic rotation.

\begin{figure}[t]
\centering
\includegraphics[width=0.72\columnwidth]{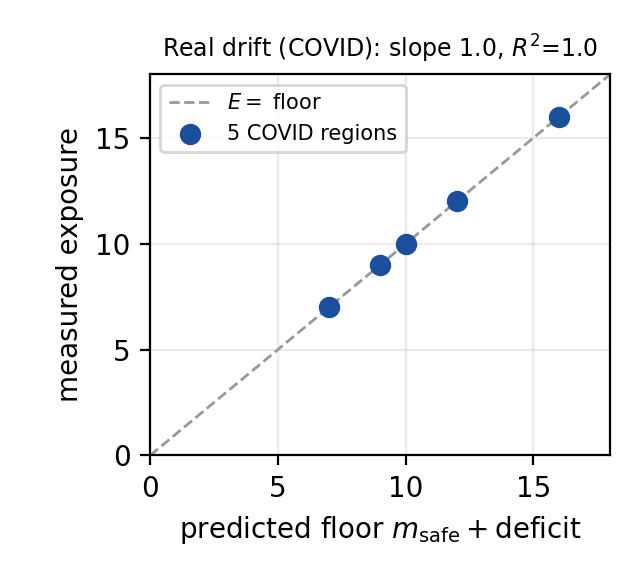}
\Description{Measured exposure equals the predicted floor across five COVID-emergent regions, slope one, R-squared one.}
\caption{The exposure law under real temporal drift. Replaying the Wikipedia
clickstream across the COVID-19 emergence (Nov--Dec 2019 $\to$ Mar--Apr 2020), the five
organically-emerging regions obey $E=m_{\mathrm{safe}}+\mathrm{deficit}$ exactly
(slope $1.0$, $R^2{=}1.0$).}
\label{fig:drift_law}
\end{figure}

\subsection{A Real Vector-Database Prototype (pgvector)}
\label{sec:pgvector}

\begin{figure}[t]
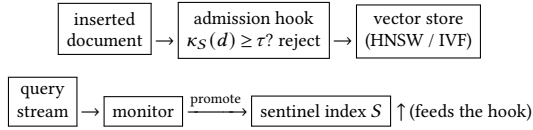

\centering
\footnotesize
\fbox{\shortstack{inserted\\document}}\,$\rightarrow$\,%
\fbox{\shortstack{admission hook\\$\kappa_S(d)\!\ge\!\tau$? reject}}\,$\rightarrow$\,%
\fbox{\shortstack{vector store\\(HNSW / IVF)}}

\vspace{2mm}
\fbox{\shortstack{query\\stream}}\,$\rightarrow$\,\fbox{monitor}\,%
$\xrightarrow{\text{promote}}$\,\fbox{sentinel index $S$}\;$\uparrow$\,(feeds the hook)
\Description{Block diagram of the ingestion path: an inserted document passes through an
admission hook that rejects it when its sentinel reverse-kNN count reaches tau, before
entering the HNSW or IVF vector store; separately, the query stream feeds a monitor that
promotes sentinels into the sentinel index, which in turn feeds the admission hook.}
\caption{The admission gate as an ingestion-path hook. Each inserted document is scored
($\kappa_S$) against the sentinel index before entering the store; a monitor maintains
that index by promoting sentinels from the drifting query stream.}
\label{fig:arch}
\end{figure}

To measure the mechanism in a real system rather than a simulation, we implement the
admission gate as an ingestion-path hook (Fig.~\ref{fig:arch}) in PostgreSQL/\texttt{pgvector}: a $500$k-vector
store behind a pgvector HNSW index, with $\kappa_S$ evaluated at insert against sentinel
thresholds precomputed from real pgvector top-$k$ queries. Measured end to end, a
pgvector insert costs $11.6$\,ms (median; $86$ docs/s) and the $\kappa_S$ admission
check adds $0.038$\,ms---a \textbf{$0.33\%$ ingest tax} at $|S|=5000$ sentinels. This
matches the analysis (\S\ref{sec:ann_ingest}): the check is $O(|S|\,d)$ and
corpus-independent, so its relative cost is set by how expensive the underlying insert
is---$0.07\%$ of a $46$\,ms 8.8M-faiss insert, $0.33\%$ of an $11.6$\,ms pgvector
insert. In both a research index and a production vector database the admission check is
nearly free; the binding cost is the exposure the approximate index imposes, not the
check.

\subsection{Scaling with Corpus Size}
\label{sec:scalability}

Building the HNSW store from $100$k to $8.8$M vectors (Fig.~\ref{fig:scalability},
Table~\ref{tab:memory}), per-insert latency grows with the corpus---$1.6$\,ms at
$100$k to $46$\,ms at $8.8$M---and index memory grows linearly at $\approx4.3$\,GB per
million vectors. The admission check is by contrast \emph{flat} in corpus size:
$O(|S|\,d)$, $0.04$\,ms regardless of $N$, and the sentinel index adds a fixed
$|S|\cdot d$ (here $20$\,MB at $|S|{=}5000$) independent of $N$. The maintenance
overhead that scales is the store's, not the admission index's, so the gate's relative
cost \emph{shrinks} as the corpus grows---from a few percent at $100$k to $<\!0.1\%$ at
$8.8$M. Table~\ref{tab:overhead} breaks down the four ingest-path operations at $8.8$M,
and Table~\ref{tab:sensitivity} the sensitivity to the admission threshold $\tau$.

\begin{figure}[t]
\centering
\includegraphics[width=0.72\columnwidth]{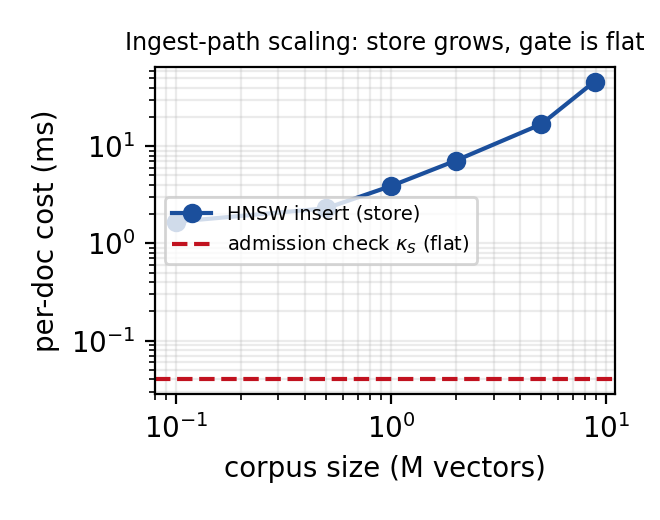}
\Description{Per-insert latency grows with corpus size on a log-log plot while the admission check stays flat and orders of magnitude below.}
\caption{Ingest-path scaling. Per-insert HNSW latency grows with corpus size while the
$O(|S|d)$ admission check stays flat and orders of magnitude below it.}
\label{fig:scalability}
\end{figure}

\begin{table}[t]
\centering
\caption{Store scaling (HNSW, $M{=}16$). Index memory grows linearly in corpus size;
the sentinel admission index adds a fixed $|S|\cdot d$ ($20$\,MB at $|S|{=}5000$),
independent of $N$.}
\label{tab:memory}
\small
\begin{tabular}{@{}rrrr@{}}
\toprule
corpus $N$ & index memory & insert (ms) & docs/s \\
\midrule
$100$k & $0.42$\,GB & $1.6$  & $600$ \\
$1$M   & $4.2$\,GB  & $3.9$  & $257$ \\
$2$M   & $8.5$\,GB  & $7.1$  & $141$ \\
$5$M   & $21$\,GB   & $16.8$ & $59$  \\
$8.8$M & $37$\,GB   & $46$   & $22$  \\
\bottomrule
\end{tabular}
\end{table}

\begin{table}[t]
\centering
\caption{Ingest-path overhead at $8.8$M, $|S|{=}5000$. The insert is the store's; the
admission index adds only the reverse-$k$NN check and, at a chosen cadence, revocation
re-evaluation---both far below the insert they ride on.}
\label{tab:overhead}
\small
\begin{tabular}{@{}lr@{}}
\toprule
ingest-path operation & per-doc cost \\
\midrule
HNSW insert (store)              & $46$\,ms \\
admission check $\kappa_S$        & $0.04$\,ms \\
promotion (maintenance update)    & one sentinel add \\
revocation re-evaluation (batched) & $0.02$\,ms/doc \\
\bottomrule
\end{tabular}
\end{table}

\begin{table}[t]
\centering
\caption{Sensitivity to the admission threshold $\tau$ ($m_{\mathrm{safe}}{=}5$,
$c_0{=}0$, $r{=}1$). The exposure floor grows linearly with $\tau$, trading tolerated
exposure against the benign false-positive rate; the per-insert admission cost is
$O(|S|d)$ and independent of $\tau$.}
\label{tab:sensitivity}
\small
\begin{tabular}{@{}rrrr@{}}
\toprule
threshold $\tau$ & exposure floor $E$ & admission check & sentinel memory \\
\midrule
$8$  & $13$ & $0.04$\,ms & $20$\,MB \\
$11$ & $16$ & $0.04$\,ms & $20$\,MB \\
$14$ & $19$ & $0.04$\,ms & $20$\,MB \\
$17$ & $22$ & $0.04$\,ms & $20$\,MB \\
$20$ & $25$ & $0.04$\,ms & $20$\,MB \\
\bottomrule
\end{tabular}
\end{table}

\subsection{Using the Bound to Provision the System}
\label{sec:provision}

Read as a design equation, the bound provisions the defense. To hold a region's
exposure at a target $E^\ast$, the promotion rate must satisfy
$r\ge\lceil\mathrm{deficit}/(E^\ast-m_{\mathrm{safe}})\rceil$. Emerging regions have
heterogeneous deficits (the five real COVID regions span $2$--$11$), so a naive
fixed-rate policy either misses the target on high-deficit regions or, sized for the
worst case, over-promotes the easy ones. Bound-guided provisioning sets $r$ per region
from its measured deficit: it meets the exposure target on \emph{every} region, where a
naive $r{=}1$ policy misses it on up to $80\%$ of them, and does so at up to $24\%$ less
sentinel-promotion churn than a conservative worst-case-sized fixed rate (measured on
the five real COVID regions above). The bound is thus not only a limit but a
configuration recipe.

\subsection{Cross-Encoder Breadth}
\label{sec:encoders}

The floor and the blind spot are properties of the count-versus-recall interaction, not
of a particular embedding. Table~\ref{tab:encoders} reruns the recall sweep on two
further encoders---GTE-large (1024-d) and MiniLM (384-d)---over a 1M-vector store. Both
reproduce the floor and the containment failure; on GTE the blind spot is in fact more
severe (IVF-Flat fails up to recall $0.78$, and IVF-PQ never contains at any
\texttt{nprobe}), and on MiniLM it recurs at the recall-calibrated operating point. With
BGE-1024 and E5-768 (\S\ref{sec:phase}), the mechanism is confirmed across four
embedding families spanning three dimensionalities, and fixB (\S\ref{sec:mitigation})
restores containment on each.

\begin{table}[t]
\centering
\caption{Cross-encoder breadth. The recall blind spot (highest index recall at which
containment still fails) recurs on every encoder; deficits and off-axis angles vary with
each encoder's geometry, the mechanism does not.}
\label{tab:encoders}
\small
\begin{tabular}{@{}lcccc@{}}
\toprule
encoder (dim) & store & deficit & blind-spot recall & fixB \\
\midrule
BGE (1024)   & 8.8M & 14 & $\le0.46$ (IVF-Flat) & restores \\
GTE (1024)   & 1M   & 6  & $\le0.78$ (IVF-Flat) & restores \\
MiniLM (384) & 1M   & 8  & $0.35$ (IVF-Flat)    & restores \\
\bottomrule
\end{tabular}
\end{table}